\documentclass[journal]{IEEEtran}
\usepackage{amsthm}

\newtheorem{coro}{Corollary}
\newtheorem{lema}{Lemma}

\usepackage[linesnumbered,ruled]{algorithm2e}
\ifCLASSINFOpdf
   \usepackage[pdftex]{graphicx}
   \graphicspath{{../pdf/}{../jpeg/}}
   \DeclareGraphicsExtensions{.pdf,.jpeg,.png}
\else
   \usepackage[dvips]{graphicx}
   \graphicspath{{../eps/}}
   \DeclareGraphicsExtensions{.eps}
\fi
\usepackage{pgfplots}
\usepackage{subfigure}
\usepackage{caption}
\usepackage{caption}
\usepackage{pseudocode}
\usepackage{tabularx}

\usepackage[cmex10]{amsmath}

\newtheorem{theorem}{Theorem}
\newtheoremstyle{case}{}{}{}{}{}{:}{ }{}
\theoremstyle{case}
\newtheorem{case}{Case}
\begin{document}

%
% paper title
% Titles are generally capitalized except for words such as a, an, and, as,
% at, but, by, for, in, nor, of, on, or, the, to and up, which are usually
% not capitalized unless they are the first or last word of the title.
% Linebreaks \\ can be used within to get better formatting as desired.
% Do not put math or special symbols in the title.
%\title{Analysis of Externalities in Social Cloud\\ {\large Extended Abstract}}
\title{Data Backup Network Formation with Heterogeneous Agents \\ {\large (Extended Abstract)}}
%
%
% author names and IEEE memberships
% note positions of commas and nonbreaking spaces ( ~ ) LaTeX will not break
% a structure at a ~ so this keeps an author's name from being broken across
% two lines.
% use \thanks{} to gain access to the first footnote area
% a separate \thanks must be used for each paragraph as LaTeX2e's \thanks
% was not built to handle multiple paragraphs
%

\author{Harshit~Jain$^{a}$, Guduru~Sai~Teja$^{a}$, Pramod~Mane$^{a}$, Kapil~Ahuja$^{a}$ and Nagarajan~Krishnamurthy$^{b}$% <-this % stops a space
\thanks{$^{a}$Authors are with the IIT Indore. Corresponding author, Kapil Ahuja's, email: \{kahuja\}@iiti.ac.in}
\thanks{$^{b}$Author is with the IIM Indore.}}% <-this % stops a space
\maketitle

% As a general rule, do not put math, special symbols or citations
% in the abstract or keywords.
\begin{abstract}
Social storage systems \cite{friendstore, hybridfriendbox, f2femepricalstudy} are becoming increasingly popular compared to the existing data backup systems like local, centralized and P2P systems. An endogenously built symmetric social storage model and its aspects like the utility of each agent, bilateral stability, contentment, and efficiency have been extensively discussed in \cite{Mane2017-social-storage}. We include heterogeneity in this model by using the concept of Social Range Matrix from \cite{stefan-social-range-matrix}. 

Now, each agent is concerned about its perceived utility, which is a linear combination of its utility as well as others utilities (depending upon whether the pair are friends, enemies or do no™t care about each other). We derive conditions when two agents may want to add or delete a link, and provide an algorithm that checks if a bilaterally stable network is possible or not. Finally, we take some special Social Range Matrices and prove that under certain conditions on network parameters, a bilaterally stable network is unique.
\end{abstract}

% Note that keywords are not normally used for peer-review papers.
% For peer review papers, you can put extra information on the cover
% page as needed:
% \ifCLASSOPTIONpeerreview
% \begin{center} \bfseries EDICS Category: 3-BBND \end{center}
% \fi
%
% For peerreview papers, this IEEEtran command inserts a page break and
% creates the second title. It will be ignored for other modes.
\IEEEpeerreviewmaketitle

\section{Introduction and Motivation}
In recent years, social storage (friend-to-friend) is emerging as an alternative to local (e.g., hard disc), online (e.g., cloud), and Peer-to-Peer data backup systems. Social storage allows users to store their data on their social relatives' (e.g., friends, colleagues, family, etc.) storage devices. Researchers believe that social relationships (which are at the core of social storage) help to mitigate issues like data availability, reliability, and security.

Initial work in this field has primarily focused on developing techniques to exogenously build social storage systems, and performing Quality of Services (QoS) analysis in terms of data reliability and availability in these systems. A most recent study in \cite{Mane2017-social-storage} focuses on explicit data backup partner selection, where agents themselves select their partners. This selection is studied in a strategic setting, and eventually builds a social storage network. They model a utility function, which reveals the benefit an individual receives in a social storage network. Further, they analyze the network by using bilateral stability as a solution concept, where no pair of agents add or delete a link without their mutual consent. 

There are several advantages of this approach. First, this approach makes it possible to incorporate user's strategic behavior (in terms of with whom user wants to form social connections and with whom it does not). Second, this approach helps us to predict the following: which network is likely to emerge; which one is stable (in which no individual has incentives to alter the structure of the network either by forming new or deleting existing social connections); and which one is the best (contended and/ or efficient) from all the participants€™ point of view.

However, the above strategic setting does not consider heterogeneous behavior (e.g., selfless, selfish, etc.) in the network formation. In this paper, we focus on this aspect. This way, we make the model close to a real-world scenario. Although, doing so makes it challenging to deal with the model and as well as predict its outcome.

To achieve the above, we incorporate the concept of Social Range Matrix introduced in \cite{stefan-social-range-matrix} while exploring different social relationships between agents. We modify the utility of agents in the social storage model discussed in \cite{Mane2017-social-storage}. Further, we revisit the results regarding bilateral stability of such a model.

%In this paper, we revisit the results discussed in \cite{Mane2017-social-storage}. If agents in the network do not care about each other utility then we have similar results. In their study, they found that in a bilateral stable network each agent has same neighbourhood size, however, in our setting, it is not true. They show that there always exists a bilaterally stable network, however, in our setting it is also not true, in some cases we do not have a bilateral stable network.

\section{Model}
Our model is composed of four components. First is symmetric social storage. Second is the utility (as derived in \cite{Mane2017-social-storage}) of each agent. Third is Social Range Matrix, which captures the social relationship between these agents. Fourth is the perceived utility obtained by combining the second and the third components.  

A symmetric social network $g$ is a data backup network consisting of $N$ number of agents and a set of links connecting these agents. A link $\langle ij \rangle \in g$ represents that the agents $i$ and $j$ are data backup partners, who store their data on each other's shared storage space. In $g$, pairs of agents share an equal amount of storage space. Agents perform two actions, link addition (represented by $g+\langle ij\rangle$) and deletion (represented by $g-\langle ij\rangle$).

The utility function in \cite{Mane2017-social-storage} reveals the cost and benefit that each agent $i$ receives in $g$ and is given by
$$u_i (g) = \beta*(1-\lambda^{n_i(g)})-c*n_i(g).$$
This utility function $u_i (g)$ consists of the following parameters: the agent $i$'s  neighborhood size $n_i(g)$; the benefit $\beta$ associated with data; the cost $c$ that the agent $i$ incurs to maintain its neighbours; and the probability of disk failure $\lambda$. Note that $c, \beta$, and $\lambda$ lie between $0$ and $1$. The utility function is a combination of two objectives for each agent $i$. The first is to minimize the total cost of the links, which is $c*n_i(g)$. The second is to maximize the expected data backup benefit, which is $\beta*(1-\lambda^{n_i(g)})$. Note that from now onwards we will use $n_i$ to represent neighbourhood size of agent $i$ in $g$.

In our model, we consider three type of agents. First, where an agent helps to maximize other's utility (friend). Second, where an agent aims to decrease other's utility (enemy). Third, where an agent does not care about other's utility (neutral). Based upon \cite{stefan-social-range-matrix}, we represent the above kind of social relationships between pairs of agents in a Social Range Matrix $F$. Here, each element  $f_{ij}$ denotes the above discussed social relationship between $i$ and $j$. That is, 
\[   
f_{ij}  
     \begin{cases}
       >\text{0} &\quad\text{$i$ and $j$ are friend of each other,}\\
      < \text{0} &\quad\text{$i$ and $j$ are enemy of each other, and  }\\
       =\text{0} &\quad\text{$i$ and $j$ do not care about each other.}\\
     \end{cases}
\]
\centerline{Also, $j \in$ the set of all agents including $i$ and $f_{ii} > 0$.}

Many such interesting matrices are possible, for example, a matrix of all ones mean that all pair of agents are friends, and a matrix of all zeroes except the diagonal elements means weak social links.

In our paper, for the sake of simplicity, we use only three values for $f_{ij}$. That is, 1, -1, and 0 indicating friend, enemy, neutral, respectively. We also consider that agents give more importance to other agents utilities than their own utility, i.e. $ f_{ii} < |f_{ij}|$  for all $j$. 

Now, agents have different social relationships with others, and hence, their utility not only depends on the structure of the network but also their social relationships with others. We define the new utility (of agent $i$ in $g$) as the perceived utility, which is defined as follows: 
\begin{align}
\widetilde{u_i}(g) = &\sum_{j}^{ }f_{ij} u_j(g).
\end{align}

In this setting, each agent's objective is to maximize its perceived utility (which takes care of utilities of other agents). Thus, the optimization problem is\\
$$max\ (\widetilde{u_i}(g)).$$

\section{Condition for link addition and deletion}
Link addition between an agent $i$ and an agent $j$ occurs only when the perceived utilities of both the agents increases, i.e.
\begin{align}
\widetilde{u_i}(g+\langle ij \rangle) &> \widetilde{u_i}(g)\hspace{0.1 cm}\\
and\nonumber\\
\widetilde{u_j}(g+\langle ij \rangle) &> \widetilde{u_j}(g).
\end{align}
For all agents $k$, except $i$ and $j$, the neighbourhood size ($n_k$) remains constant so their utility is the same and cancels out. Simplifying (2) and (3) with (1), we get that link addition happens when 
%Solving with utility in (1) we get\\
%\centerline{$f_{ii}[\beta(\lambda^{n_i}-\lambda^{n_i+1})-c]+f_{ij}[\beta(\lambda^{n_j} - \lambda^{n_j+1})-c] > 0$}
%\vspace{0.1 cm} Similarly we solve (2) with utility in (1). Finally, the simplified inequalities are given below.
\begin{align}
f_{ii}{\lambda}^{n_i}+f_{ij}{\lambda}^{n_j}>\frac{(f_{ii}+f_{ij})*c}{(1-\lambda)*\beta} & \quad and \\
f_{ii}{\lambda}^{n_j}+f_{ij}{\lambda}^{n_i}>\frac{(f_{ii}+f_{ij})*c}{(1-\lambda)*\beta}. &
\end{align}
Similarly, link deletion happens when
\begin{align}
\widetilde{u_i}(g-\langle ij \rangle) &> \widetilde{u_i}(g)\hspace{0.1 cm}\nonumber \\
and\nonumber\\
\widetilde{u_j}(g-\langle ij \rangle) &> \widetilde{u_j}(g).\nonumber
\end{align}
Thus, equivalent to (4) and (5) we get link deletion conditions as below.
\begin{align}
f_{ii}{\lambda}^{n_i}+f_{ij}{\lambda}^{n_j}&<\frac{(f_{ii}+f_{ij})*c*\lambda}{(1-\lambda)*\beta} & and\\
f_{ii}{\lambda}^{n_j}+f_{ij}{\lambda}^{n_i}&<\frac{(f_{ii}+f_{ij})*c*\lambda}{(1-\lambda)*\beta}. &
\end{align}

\section{Sufficiency condition for link addition}
\begin{theorem}
For any two agents $i$ and $j$, let $t_{1}$ = max\ ($n_{i}$, $n_{j}$), $t_{2}$ = min\ ($n_{i}$, $n_{j}$) and $f_{ij} >$ 0. If  $t_{1} < \frac{\left|\ln\left(\frac{c}{\left(1-\lambda\right)*\beta}\right)\right|}{\left|\ln \lambda\right|}$, the link addition conditions (4) and (5) will be true, and hence, a link will be formed between $i$ and $j$.
\end{theorem}
\begin{proof}
\begin{case}
If $ n_i  > n_j,$ then $t_1 = n_i$ and $t_2 = n_j.$ Thus, 
\begin{align}
& {\lambda}^{n_i}& <{\lambda}^{n_j} \hspace{1 cm}  (since\hspace{0.1 cm} 0<\lambda<1) \nonumber \\
or \quad & f_{ii}{\lambda}^{n_i}& <f_{ii}{\lambda}^{n_j} \hspace{1 cm} (since\hspace{0.1 cm} f_{ii}>0). \nonumber\\
Also, \quad & f_{ij}{\lambda}^{n_i}& <f_{ij}{\lambda}^{n_j}\hspace{1 cm} (since\hspace{0.1 cm} f_{ij}>0).\nonumber
\end{align}
Combining above two we get 
\begin{align}
(f_{ii}+f_{ij})\lambda^{n_i} < f_{ii}\lambda^{n_i} &+ f_{ij}\lambda^{n_j} < (f_{ii} + f_{ij})\lambda^{n_j}\\
&and\nonumber\\
(f_{ii}+f_{ij})\lambda^{n_i} < f_{ii}\lambda^{n_j} &+ f_{ij}\lambda^{n_i} < (f_{ii} + f_{ij})\lambda^{n_j}.
\end{align}
Using (4)--(5) in (8)--(9) we get sufficiency condition for link addition as below. 
\begin{align}
\hspace{0.1 cm}(f_{ii}+f_{ij})\lambda^{n_i} &> \frac{(f_{ii}+f_{ij})*c}{(1-\lambda)*\beta}.
\end{align}
On solving (10), we get that link between $i$, $j$ will happen if 
\begin{align}
n_i < & \frac{\left|\ln\left(\frac{c}{\left(1-\lambda\right)*\beta}\right)\right|}{\left|\ln \lambda\right|}. \nonumber
\end{align}
\end{case}
\begin{case}
If $n_i \le n_j$\\
Since (4) and (5) are symmetric, interchange $i$ and $j$.
\end{case}
\end{proof}
Similarly, we can derive sufficiency conditions when $f_{ij} < 0$ above as well as for link deletion. Algorithm 1 lists the steps for reaching a bilaterally stable network. That is, when no agent has any incentive to add or delete a link.
\begin{algorithm}
\footnotesize
    \SetKwInOut{Input}{Input}
    \SetKwInOut{Comment}{Comment}

    %\underline{function Euclid} $(a,b)$\;
    \Input{$c,\lambda, \beta, F, starting$ $network, flag = 1$}
\Comment{$i$ and $j$ are agents, flag = 1 means network is not bilaterally stable}	
    %\Output{$\gcd(a,b)$}
\While{flag $==$ 1}{
flag $=$ 0\\
\For{$ i = 1$ to n}{
\For{$j = 1$ to n}{
\If {$i \neq j$}
{
\If { link is absent between $<i,j>$}
{
 Check link addition conditions (4) and (5) for $i$, $j$ and add link if they are true.\\
flag $=$ 1
}
\If { link is present between $<i,j>$}
{
  Check link deletion conditions (6) and (7) for $i$, $j$ and delete link if they are true.\\
flag $=$ 1
}
}

}
}
}
    \caption{Pseudo code to arrive at a bilaterally stable network.}
\end{algorithm}
\section{Case study}
Consider five agents as follows: $a, b, c, d$, and $e$. Assume that social relationships between these agents is captured in the Social Range Matrix $F$ given in Table \ref{my-label}. As an example, the way to read this matrix is as follows: $a$ is a friend of $b$ and $d$, while it is an enemy of $c$ and $e$. 

\begin{table}[!h]
\centering
\small
\caption{Social Range Matrix $F$}
\label{my-label}
\begin{tabular}{llllll}
    & $a$        & $b$        & $c$        & $d$        & $e$        \\
$a$ & $\epsilon$ & 1          & -1         & 1          & -1         \\
$b$ & 1          & $\epsilon$ & 1          & -1         & -1         \\
$c$ & -1         & 1          & $\epsilon$ & -1         & 1          \\
$d$ & 1          & -1         & -1         & $\epsilon$ & 1          \\
$e$ & -1         & -1         & 1          & 1          & $\epsilon$
\end{tabular}
\end{table}

\begin{figure}
\footnotesize
\begin{center}
      \includegraphics[width=3.5 cm , height=2.5cm]{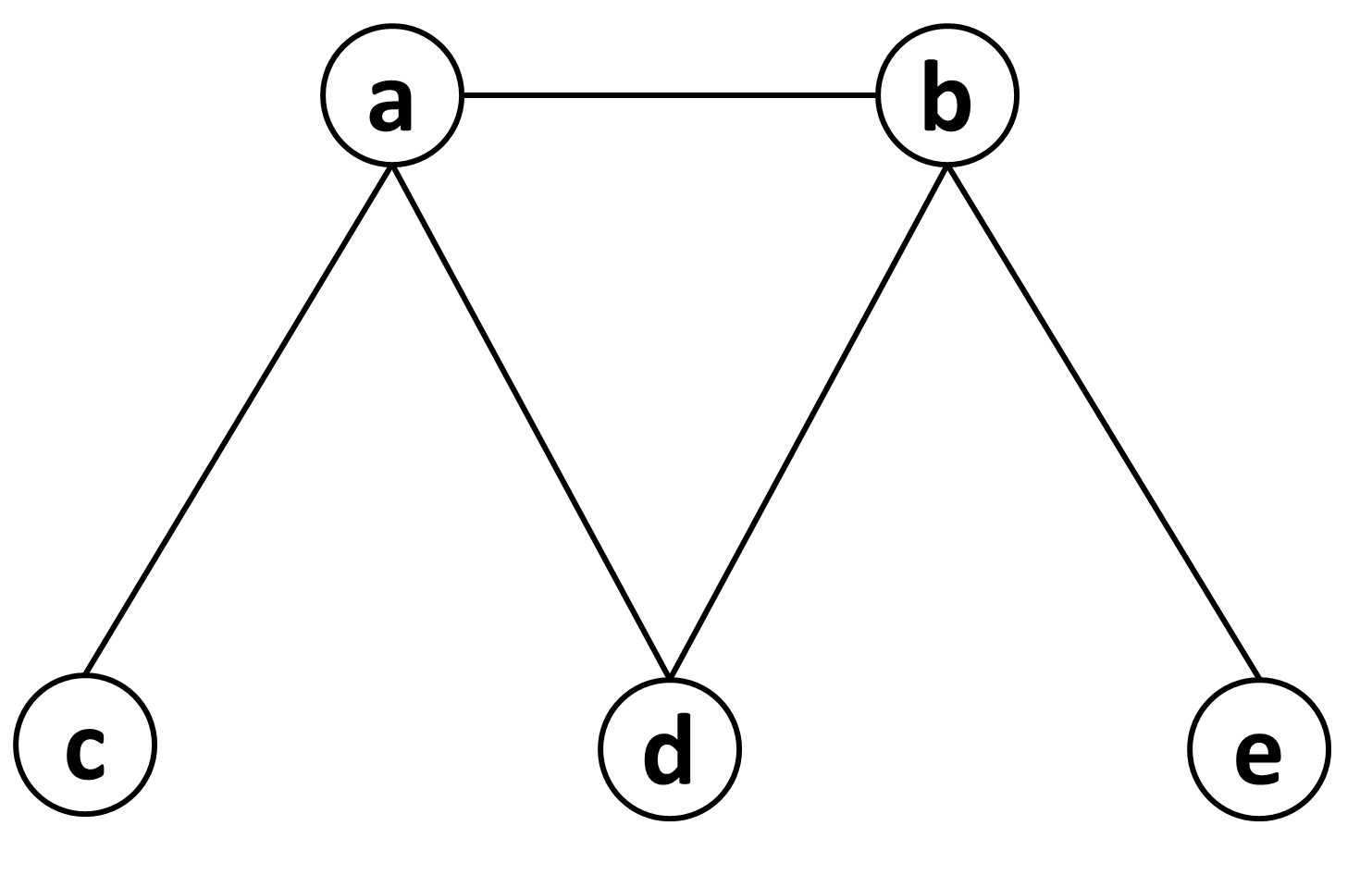}
\end{center}
  \caption{Bilaterally Stable Social Storage Network $g'$.}
  \label{fig:boat1}
\end{figure}

Let us also assume that $c = 0.01, \lambda= 0.2, \epsilon=0.1$, and $\beta=0.1$. Initially all agents are isolated (i.e., the starting network is a null network). We follow the procedure described in Algorithm 1 and obtain the bilaterally stable network $g'$ shown in Fig. \ref{fig:boat1}. With the same set of parameters and $F$, by starting from a random non-null network, we obtain a different bilaterally stable network than in Fig. \ref{fig:boat1}. This implies that bilaterally stable networks are not necessarily unique.

By looking at the ratio of $\frac{c}{\beta}$, where $0<c, \beta<1$, $\lambda= 0.2$, and above $F$, we found that for $0.044 < \frac{c}{\beta} < 0.089$, Algorithm 1 runs into infinite loop implying that no bilaterally stable network is possible. Interestingly, for all other values of $\frac{c}{\beta}$ we found that at least one bilaterally stable network exists. 

If $c > \beta$, $\lambda= 0.2$, and above $F$, then we get a unique bilaterally stable network in which all pair of agents are enemies.

\section{Unique Bilaterally Stable Networks}
\begin{lema}
If $\frac{c}{\beta} < (1-\lambda)\lambda^{N-2}$ and if all agents are friends of each other (i.e. $f_{ij}=1$), then a complete network is the unique bilaterally stable network.\\
\end{lema}
\begin{proof}
First, consider the link addition conditions stated in (4) and (5). We can observe that, as $n_i$ and $n_j$ increases, the L.H.S of (4) and (5) decreases. This implies that agents have an incentive to increase their neighborhood size.

If the link addition conditions (4) and (5) are true for $n_i = n_j = N-2$ (both agents have the neighborhood size one less than the maximum possible size)\footnote{Neighborhood size of $N-1$ is a complete network, and hence, is not included in the analysis here.}, then they will be true for all values of $n_i$, $n_j <$ N-2 (if the smaller value of the L.H.S is greater than the R.H.S., bigger values will be greater).

Solving (4) -- (5) with $n_i = n_j = N-2$, i.e.
\begin{align*}
{f_{ii}{\lambda}^{N-2}+f_{ij}{\lambda}^{N-2}>\frac{(f_{ii}+f_{ij})*c}{(1-\lambda)*\beta}} 
\end{align*}
gives
$${\frac{c}{\beta} < (1-\lambda)\lambda^{N-2}}.$$
%Now we prove 2.
%We consider the link deletion conditions (6) and (7). If min(L.H.S of (4) and (5)) $>$ R.H.S of (4) and (5) i.e. if the minimum is greater than the R.H.S, then all other values will definitely be greater. The minimum value will be when $n_i = n_j =$ N-1 (when both agents have achieved maximum neighbourhood size). So \\
%\centerline{$f_{ii}{\lambda}^{N-1}+f_{ij}{\lambda}^{N-1}>\frac{(f_{ii}+f_{ij})*c*\lambda}{(1-\lambda)*\beta}$}
%\centerline{$f_{ii}{\lambda}^{N-1}+f_{ij}{\lambda}^{N-1}>\frac{(f_{ii}+f_{ij})*c*\lambda}{(1-\lambda)*\beta}$}
%\centerline{$\frac{c}{\beta} < (1-\lambda)\lambda^{n-2}$}
A similar analysis can be done with the deletion conditions.
\end{proof}

\begin{lema}
 If $\frac{c}{\beta} > (1-\lambda)$ and if all agents are friends of each other (i.e. $f_{ij}=1$), then the empty network is the unique bilaterally stable network.
\end{lema}
\begin{proof}
Similar to Lemma 1.
\end{proof}

\begin{coro}
If $\frac{c}{\beta} >\frac{1}{1-f_{ii}} (1-\lambda)$ and if $f_{ij}$ is either 1 or -1, then the bilaterally stable network will be the one where all the pairs of enemies form links.
\end{coro}

\section{Future Work}
In this work, we have extended the social storage model proposed in \cite{Mane2017-social-storage} to include heterogeneous behavior (variety of social relationships). After that, we have analyzed this model using their solution concept of bilateral stability. The preliminary results here give tremendous insight into how a endogenously built social storage system would emerge.    

Future work involves further analyzing the stability of such networks as well as studying contentment (when everyone has maximized their utility) and efficiency (where the total utility of the network is maximized) in this new heterogeneous agent scenario.  

%\section{Future Work}
%There may be many different bilaterally stable networks possible for a given set of network parameters and a given social range matrix, depending upon the starting network and the algorithm used for network formation. We are trying to introduce the concept of an efficient network, which is a network where the sum of perceived utilities of all agents is the maximum. This analysis is useful from the perspective of an outside network manager who is concerned about the total network's utility rather than an individual agents.\\
%We were able to derive sufficiency conditions for the link addition and deletion equations in terms of the neighborhood sizes. Finding necessary conditions could be another direction of future work.

\bibliographystyle{IEEEtran}
\bibliography{ComsNets-2018}

% Generated by IEEEtran.bst, version: 1.13 (2008/09/30)
\begin{thebibliography}{1}
\providecommand{\url}[1]{#1}
\csname url@samestyle\endcsname
\providecommand{\newblock}{\relax}
\providecommand{\bibinfo}[2]{#2}
\providecommand{\BIBentrySTDinterwordspacing}{\spaceskip=0pt\relax}
\providecommand{\BIBentryALTinterwordstretchfactor}{4}
\providecommand{\BIBentryALTinterwordspacing}{\spaceskip=\fontdimen2\font plus
\BIBentryALTinterwordstretchfactor\fontdimen3\font minus
  \fontdimen4\font\relax}
\providecommand{\BIBforeignlanguage}[2]{{%
\expandafter\ifx\csname l@#1\endcsname\relax
\typeout{** WARNING: IEEEtran.bst: No hyphenation pattern has been}%
\typeout{** loaded for the language `#1'. Using the pattern for}%
\typeout{** the default language instead.}%
\else
\language=\csname l@#1\endcsname
\fi
#2}}
\providecommand{\BIBdecl}{\relax}
\BIBdecl

\bibitem{friendstore}
N.~Tran, F.~Chiang, and J.~Li, ``Efficient cooperative backup with
  decentralized trust management,'' \emph{ACM Transactions on Storage}, vol.~8,
  no.~3, pp. 8:1--8:25, 2012.

\bibitem{hybridfriendbox}
R.~Gracia-Tinedo, M.~S{\'a}nchez-Artigas, A.~Moreno-Mart{\'i}nez, and
  P.~Garcia-Lopez, ``Friendbox: A hybrid {F2F} personal storage application,''
  in \emph{12th IEEE International Conference on Peer-to-Peer Computing
  (P2P)}.\hskip 1em plus 0.5em minus 0.4em\relax IEEE, 2012, pp. 131--138.

\bibitem{f2femepricalstudy}
R.~Sharma, A.~Datta, M.~DeH'Amico, and P.~Michiardi, ``An empirical study of
  availability in friend-to-friend storage systems,'' in \emph{IEEE
  International Conference on Peer-to-Peer Computing (P2P)}, 2011, pp.
  348--351.

\bibitem{Mane2017-social-storage}
\BIBentryALTinterwordspacing
P.~Mane, K.~Ahuja, and N.~Krishnamurthy, ``Unique stability point in social
  storage networks,'' 2017, 33 Pages. [Online]. Available:
  \url{https://arxiv.org/abs/1603.07689}
\BIBentrySTDinterwordspacing

\bibitem{stefan-social-range-matrix}
P.~Kuznetsov and S.~Schmid, ``Towards network games with social preferences,''
  in \emph{17th International Colloquium on Structural Information and
  Communication Complexity (SIROCCO)}, B.~Patt-Shamir and T.~Ekim, Eds.\hskip
  1em plus 0.5em minus 0.4em\relax Springer: Berlin Heidelberg, 2010, pp.
  14--28.

\end{thebibliography}

\end{document}